\documentclass[letterpaper]{article} 
\usepackage[submission]{aaai23}  
\usepackage{times}  
\usepackage{helvet}  
\usepackage{courier}  
\usepackage[hyphens]{url}  
\usepackage{graphicx} 
\usepackage{natbib}  
\usepackage{caption} 
\usepackage[noend]{algorithmic}
\usepackage[linesnumbered,ruled,vlined]{algorithm2e}

\usepackage{amsmath}
\usepackage{amsthm}
\usepackage{amssymb}
\usepackage{mathtools}
\usepackage{cleveref}
\usepackage{comment}
\usepackage[shortlabels]{enumitem}
\usepackage{xcolor}
\usepackage{booktabs}
\usepackage{pifont}
\usepackage{thmtools} 
\usepackage{thm-restate}

\newtheorem{theorem}{Theorem}[section]

\newtheorem{definition}[theorem]{Definition}

\DeclareMathOperator{\argmax}{argmax}
\DeclarePairedDelimiter{\ceil}{\lceil}{\rceil}
\DeclarePairedDelimiter{\floor}{\lfloor}{\rfloor}
\let\epsilon\varepsilon
\let\eps\varepsilon
\newcommand{\distribution}{\mathcal{D}}
\newcommand{\set}[1]{\{#1\}}
\newcommand{\modulus}[1]{\vert #1 \vert}
\newcommand{\Prob}[2]{\Pr_{#2}\left[\,#1\,\right]}
\newcommand{\Wthreshold}{\ensuremath{\floor*{\log_2{\frac{1}{\varepsilon}}}}}
\newcommand{\Ni}{\mathcal{N}_i}
\newcommand{\RR}{\ensuremath{\mathbb{R}}}
\newcommand{\hclass}{\ensuremath{\mathcal{H}}}
\newcommand{\W}{\ensuremath{\mathcal{W}}}
\newcommand{\sample}{\ensuremath{\mathcal{S}}}

\newcommand{\cmark}{\ding{51}}%
\newcommand{\xmark}{\ding{55}}%
\newcommand{\learnkDL}{\textsc{Learn-}k\textsc{-DL}}
\newcommand{\learnWGames}{\textsc{Learn-}\W\textsc{-games}}

\makeatletter
\newcommand{\proofpart}[1]{%
  \par
  \addvspace{\smallskipamount}%
  \noindent\emph{#1} \hspace{2pt}
}
\makeatother

\title{PAC learning and stabilizing Hedonic Games: towards a unifying approach.}
\author {
    Simone Fioravanti,\textsuperscript{\rm 1}
    Michele Flammini,\textsuperscript{\rm 1}
    Bojana Kodric,\textsuperscript{\rm 2, \rm 1}
    Giovanna Varricchio \textsuperscript{\rm 3}
}
\affiliations {
    \textsuperscript{\rm 1} Gran Sasso Science Institute (GSSI), L'Aquila, Italy\\
    \textsuperscript{\rm 2} Ca' Foscari University of Venice, Venice, Italy\\
    \textsuperscript{\rm 3} Goethe-Universität, Frankfurt am Main, Germany\\
    simone.fioravanti@gssi.it,
    michele.flammini@gssi.it,
    bojana.kodric@unive.it,
    varricchio@em.uni-frankfurt.de
}

\begin{document}

\maketitle

\begin{abstract}
We study PAC learnability and PAC stabilizability of Hedonic Games (HGs), i.e., efficiently inferring preferences or core-stable partitions from samples.
 We first expand the known learnability/stabilizability landscape for some of the most prominent HGs classes, providing results for Friends and Enemies Games, Bottom Responsive, and Anonymous HGs. 
 Then, having a broader view in mind, we attempt to shed light on the structural properties leading to learnability/stabilizability, or lack thereof, for specific HGs classes.
 Along this path, we focus on the fully expressive Hedonic Coalition Nets representation of HGs. We identify two sets of conditions that lead to efficient learnability, and which encompass all of the known positive learnability results.
 On the side of stability, we reveal that, while the freedom of choosing an ad hoc adversarial distribution is the most obvious hurdle to achieving PAC stability, it is not the only one. First, we show a distribution independent necessary condition for PAC stability.
Then, we focus on $\W$-games, where players have individual preferences over other players and evaluate coalitions based on the least preferred member.
We prove that these games are PAC stabilizable under the class of bounded distributions, which assign positive probability mass to all coalitions.
Finally, we discuss why such a result is not easily extendable to other HGs classes even in this promising scenario. Namely, we establish a purely computational property necessary for achieving PAC stability.
\end{abstract}

\section{Introduction}
Hedonic Games (HGs)~\cite{Dreze80} are a formal model for describing selfish individuals gathering together in order to form coalitions. Both HGs and general coalition formation games attracted considerable research attention in the last years due to their applicability to multi-agent environments.
Solution concepts for HGs are usually in the form of agent partitions with some suitable properties. The one we consider in this paper is core stability. A partition is said to be \emph{core-stable} (or in the core) if there exists no subset of players that could regroup into a so-called \emph{core-blocking} coalition, which is preferred by all of them.

The usual assumption when considering any solution concept is that the preferences of the agents are fully known, which is arguably unrealistic.
Could we instead efficiently infer the whole game structure, or even directly learn solution concepts, while having only partial knowledge of the preferences?
Questions of this kind are naturally captured by the \emph{probably approximately correct (PAC) learning} framework~\cite{Valiant84}, which formalizes the problem of learning a target concept from a limited number of samples from any possible unknown but fixed distribution.

\citet{SliwinskiZ17} were the first to leverage the PAC framework to study the problem of learning HGs preferences and core-stable partitions from samples. In particular, they define PAC stabilizability of a HGs class as the property of being able to, upon seeing a limited number of samples, either report that the core is empty or propose a partition that is unlikely to be core-blocked by further coalitions sampled from the same distribution.
In a recent paper, \citet{HG_Knesset} apply the notion of PAC stabilizability of HGs in the context of political coalition formation. In particular, they use the publicly available Israeli parliament voting data to fit a Friends Appreciation HG, and compare the actual political parties of the voters to the PAC-stable coalitions resulting from the model. This example shows how learning concepts have the potential to create space for applications of mainly theoretical models, as HGs.

While the work of~\citet{SliwinskiZ17} and the ones that followed considered PAC learnability and stabilizability of many specific classes of HGs, the overall picture is still far from being complete. Most prominently, the characterization of the underlying general conditions explaining the existing results is missing. Furthermore, PAC stabilizability seems very hard to achieve and it is natural to wonder whether some restrictions on the PAC stability definition can yield better results.
Here, we address these questions, attempting to provide a deeper theoretical understanding of what makes HGs learnable and stabilizable.

\subsection{Our Contribution}
We first extend the knowledge on PAC learnable and PAC stabilizable classes of HGs. We start by focusing on Friends and Enemies Games, examining whether the negative results on stabilizability of Additively Separable HGs transfer to this simple subclass. By exploiting previous results and proposing an algorithm stabilizing Friends and Enemies under Enemies Aversion, we deduce that Friends and Enemies Games belong to the very few lucky HGs classes that can both be learned and stabilized. Next, we study Bottom Responsive HGs and show that while they are not efficiently learnable, they are stabilizable. Finally, we turn our attention to Anonymous HGs and show that the opposite holds here, i.e., they are efficiently learnable but not stabilizable.

After exploring specific HGs classes, we use the gained insights to follow a more general research direction, devoted to a deeper understanding of the structural properties that make HGs learnable and/or stabilizable.

We first consider the learning problem. Additively Separable, Anonymous, $\mathcal{W}$ and $\mathcal{B}$-games are all known to be learnable, and we investigate why this is the case. To this aim, we consider Hedonic Coalition Nets (HCNs), a general framework for representing HGs that is universally expressive, i.e., it can represent any HGs class. 
We identify two sets of conditions on the HCNs representation that imply efficient learnability, and as special cases explain the learnability of all of the aforementioned HGs classes.

We then turn our attention to stability. Achieving PAC stability does not seem possible for most HGs classes, and we try to find general reasons causing this fact. 
First, we show a simple necessary condition for PAC stability, abstracting the proof pattern of all the known negative results for specific HGs classes.
Then, we consider the problem of PAC stability with bounded probability distributions and prove that under this restriction it is possible to PAC stabilize $\mathcal{W}$-games, which is known not to be possible in general. Finally, we discuss why the same result cannot be easily extended to other HGs. In particular, we determine a general purely computational property 
necessary for achieving PAC stability.

Due to space limitations, all the missing proofs are deferred to the Appendix.

\subsection{Related Work}
Many works have dealt with learning game-theoretic solution concepts from data.
\citet{SliwinskiZ17} first introduced the PAC learning framework into the study of HGs. Their work was extended by~\citet{IgarashiSZ19} to tackle HGs with underlying players' interaction networks. Moreover, \citet{JhaZick20} laid further foundations for learning game-theoretic solution concepts from samples. More recently, \citet{HGnoisy} studied the problem of learning HGs with noisy preferences.

Other works have considered learning cooperative games~\cite{BalcanPZ15}, markets~\cite{LevPVZick21}, auctions~\cite{BalcanMultiItem} but also, more generally, combinatorial functions~\cite{BalcanPairwise, BalcanSubmodular}.

There is a vast body of literature on HGs. For a thorough introduction to the main concepts and results, we refer to~\citet{AzizS16}, where both all the HGs classes studied in this paper and also HCNs are discussed.

\section{Preliminaries}
Let $N$ be a set of $n$ \emph{players}. We call any non-empty subset $S \subseteq N$ a \emph{coalition} and denote by $\Ni$ the set of all coalitions which contain a given player $i \in N$. We call any coalition of size one a \emph{singleton}.
We denote by $\succsim_i$ any binary \emph{preference relation} of player $i$ over the coalitions in $\Ni$, which is reflexive, transitive, and complete.
A \emph{Hedonic Game} (HG) is then a pair $H=( N, \succsim )$, where $\succsim = (\succsim_i, \ldots, \succsim_n)$ is a \emph{preference profile}, i.e., the collection of all players' preferences.
Throughout this work we will assume that players' preferences are expressed as real numbers by means of \emph{valuation} functions $v_i$. In other words, given $S,T \in \Ni$: $v_i(S) \geq v_i(T)$ if and only if $S \succsim_i T$. We will denote by $\vec{v}=(v_1,\ldots, v_n)$ the collections of players' valuations and assume that $v_i(S)=\varnothing$ for $S\notin \Ni$.
Let $H$ be a HG and $\pi$ a \emph{coalition structure}, i.e., a partition of players into coalitions. A set $S$ is said to \emph{core-block} $\pi$ if $v_i(S) > v_i(\pi(i))$ for each $i \in S$, where $\pi(i)$ denotes the coalition containing $i$ in $\pi$. 
A coalition structure $\pi$ is \emph{core-stable} if there does not exist a core-blocking coalition $S \subseteq N$. Among the many possible solution concepts, the one we will consider in this paper is core stability, as it is the most prominent one in the PAC stability model.

\subsection{Defining Classes of Hedonic Games}
In this subsection, we provide the definitions of some HGs classes already considered from the perspective of PAC learning by~\cite{SliwinskiZ17}, that will be frequently mentioned in the sequel.
In all of these classes, for a player $i \in N$ and a coalition $S\in\Ni$, the valuation $v_i(S)$ is completely determined by the values $v_i(j)$ for $j \in S\setminus \{i\}$. More precisely, the valuation of $i$ for $S$ is equal to:
\begin{enumerate}
    \item \emph{Additively Separable}: the sum of the values of its members, i.e., $v_i(S) = \sum_{j \in S\setminus\{i\}} v_i(j)$;
    \item \emph{Fractional}: the sum of the values of its members, but normalized by the size of the coalition, i.e., $v_i(S) = \sum_{j \in S\setminus\{i\}} v_i(j)/ \modulus{S}$;
    \item \emph{$\W$-games}: the value of the worst player in the coalition;
    \item \emph{$\mathcal{B}$-games}: the value of the best player in the coalition, but coalitions of smaller size are preferred.
\end{enumerate}

\subsection{PAC Learning}
The PAC learning model, originally introduced by~\citet{Valiant84}, mathematically formalizes the process of learning a \emph{target concept} $v$ belonging to a \emph{hypothesis class} $\hclass$, by using a sample of labeled examples as input. There are many variants, which adapt to different learning paradigms. In the following, we will formally present only the one that we will use in this work.
Our aim is to learn an unknown valuation function $v:2^N \rightarrow \RR$ within a class $\hclass$, given as input $\sample = \{( S_1, v(S_1)), \ldots, (S_m, v(S_m))\}$, i.e., a collection of coalition/valuations pairs.
The distribution $\distribution$, according to which the i.i.d.\ input coalitions are sampled, is unknown, while the class $\hclass$ is determined by the HG instance one considers, e.g., if one is studying Additively Separable HGs, $\hclass$ will be the class of additively separable functions over (the other) $n-1$ players.
Starting from a sample $\sample$, learning is the process of producing a hypothesis $v^* \in \hclass$ which is as close as possible to the real $v$. 
Formally, a hypothesis $v^* \in \hclass$ is \emph{$\eps$-approximately correct} w.r.t.\ a distribution $\distribution$ over $2^N$ and a function $v \in \hclass$, if the following holds:
\begin{equation*}
    \Prob{v^*(S) \neq v(S)}{S\sim \distribution} < \eps \ .
\end{equation*}
Given $\eps, \delta >0$, class $\hclass$ is $(\eps, \delta)$ \emph{probably approximately correctly (PAC) learnable} if there exists an algorithm $\mathcal{A}$ that, for every distribution $\distribution$ over $2^N$, and any $v \in \hclass$, given a sample drawn from $\distribution$, is able to produce a hypothesis $v^*$ which is $\eps$-approximately correct with probability at least $1-\delta$. A class $\hclass$ is said to be PAC learnable if it is $(\eps, \delta)$ PAC learnable for all $\eps, \delta >0$. Furthermore, if the sample size $m$ and the running time of $\mathcal{A}$ are polynomial in $\frac{1}{\eps}, \log{\frac{1}{\delta}}$ and $n$, $\hclass$ is said to be efficiently PAC learnable.

The inherent complexity of efficiently PAC learning a concept class of real functions $\hclass$ is usually measured by the so-called \emph{pseudo-dimension} (see, e.g., \citet{AnthonyBartlett}), which is the analog of the more renowned \emph{VC-dimension}~\cite{KearnsVazirani} defined only for classes of binary functions.
In order to formally define pseudo-dimension, we first need to introduce the concept of \emph{pseudo-shattering}. Given a collection of coalition/value pairs $\sample = \{(S_1, r_1), \ldots, (S_q, r_q)\}$, we say that a class $\hclass$ can \emph{pseudo-shatter} $\sample$ if, for every possible binary labeling $l_1,\ldots, l_q$ of $\sample$, there exists a function $f \in \hclass$ such that $f(S_j) > r_j \iff l_j=1$. Intuitively, the more $\hclass$ is expressive, the bigger the sets that it can pseudo-shatter. The pseudo-dimension of $\hclass$, denoted as $P_{dim}(\hclass)$, is the size of the maximal set $\sample$ that can be pseudo-shattered by $\hclass$.

We conclude this section by reporting the theorem which bridges learning and pseudo-dimension.
\begin{theorem}[\citeauthor{AnthonyBartlett} \citeyear{AnthonyBartlett}]\label{thm:PAC_learnability}
A hypothesis class $\mathcal{H}$ with $P_{dim}(\hclass)$ polynomial in $n$ is $(\varepsilon, \delta)$ PAC learnable using $m$ samples, where $m$ is polynomial in $P_{dim}(\mathcal{H}), \frac{1}{\varepsilon}$ and $\log \frac{1}{\delta}$, by any algorithm $\mathcal{A}$ that returns a hypothesis $f^*$ consistent with the sample, i.e., $f^*(S_i)=f(S_i)$ for all $i$.
Furthermore, if $P_{dim}(\mathcal{H})$ is superpolynomial in $n$, $\mathcal{H}$ is not efficiently PAC learnable.
\end{theorem}
\subsection{PAC Stabilizing Hedonic Games}
The concept of \emph{PAC stabilizing} HGs was first introduced in~\cite{SliwinskiZ17}. A coalition structure $\pi$ is said to be \emph{$\eps$-PAC stable} under a distribution $\distribution$ if
$\Prob{S \text{ core-blocks } \pi}{S\sim \distribution} < \eps$.
A class of HGs $\hclass$ is \emph{PAC stabilizable} if there exists an algorithm $\mathcal{A}$ that for any HG in $\hclass$, any $\eps, \delta>0$, and any $\distribution$ over $2^N$, given a sample $\sample =\{( S_1, \vec{v}(S_1)), \ldots, ( S_m, \vec{v}(S_m))\}$ of coalitions drawn according to $\distribution$, produces an $\eps$-PAC stable coalition structure $\pi$ under $\distribution$ with probability at least $1-\delta$, or reports that the core is empty. If the sample size $m$ and the running time of $\mathcal{A}$ meet the same conditions required for efficient PAC learnability, we say that $\hclass$ is efficiently PAC stabilizable.
Intuitively, this concept formalizes the learnability of a solution concept for a HGs class, independently from the learnability of the class itself. 
We will rely on the following theorem in the next section.
\begin{theorem}[\citeauthor{JhaZick20} \citeyear{JhaZick20}]\label{thm:PAC_stabilizability}
A class of HGs $\mathcal{H}$ is efficiently PAC stabilizable iff there exists an algorithm that outputs a partition $\pi$ consistent with the sample, i.e., no coalition from the sample core-blocks $\pi$.
\end{theorem}

\section{Learnability and Stabilizability of New Classes of Hedonic Games}\label{sec:new_classes}
In this section we broaden the picture of learnability and stabilizability of different classes of HGs, studying the following HG classes that were not considered by previous work.

\paragraph{Friends and Enemies.}
Friends and Enemies Games have been traditionally investigated under two types of preference profiles, called \emph{Friends Appreciation} and \emph{Enemies Aversion}, where agents prefer  coalitions with a greater number of friends (and smaller number of enemies in case of ties) or with a smaller number of enemies (and greater number of friends in case of ties), respectively. 

\paragraph{Bottom Responsive.}
The \emph{bottom responsiveness} property was first defined by \citet{suzuki2010hedonic} as bottom \emph{refuseness} and then further considered by \citet{Aziz_BottomResp}, where it was renamed in analogy to a related property called \emph{top responsiveness}. Intuitively, it models pessimistic agents who rank coalitions based on sets of players that they would like to avoid. 

\begin{definition}\label{def:bottom_responsive}
For each player $i\in N$ and $S\in \mathcal{N}_i$, we define the \emph{avoid set} of player $i$ in coalition $S$ as
\[
Av(i,S)=\{S'\subseteq S : (i\in S') \land (\forall S''\subseteq S, S'\preceq_i S'') \}.
\]
A game satisfies \emph{bottom responsiveness} if for each $i\in N$ and for each pair $S,T\in \mathcal{N}_i$ the following conditions hold:
\begin{itemize}
    \item[(i)] if for each $S'\in Av(i,S)$ and for each $T'\in Av(i,T)$ it holds that if $S'\succ_i T'$, then $S\succ_i T$;
    \item[(ii)] if $Av(i,S)\cap Av(i,T)\neq \emptyset$ $\land$ $\modulus{S}\geq \modulus{T}$, then $S\succeq_i T$.
\end{itemize}
\end{definition}

In what follows, we assume a minimum a priori knowledge of the values. Namely, we assume to know $v_i(\{i\}), \forall i\in N$. A similar, yet significantly stronger, assumption was used by~\cite{SliwinskiZ17} to prove that Top Responsive HGs (i.e. HGs which satisfy top responsiveness) are efficiently PAC stabilizable.

\paragraph{Anonymous.}
A HG is said to satisfy \emph{anonimity}, as defined in~\cite{Banerjee01, Bogomolnaia02}, if $v_i(S) = v_i(T)$ for any player $i \in N$ and any $S,T\in \mathcal{N}_i$ with $|S|=|T|$, i.e., players evaluate coalitions only according to their size. 

\begin{table}\centering
\begin{tabular}{@{}lcc@{}}
\toprule
HGs class & Learnable & Stabilizable\\
\midrule
\textbf{Friends and Enemies} & & \\
Friends Appreciation & \cmark$^*$ & \cmark$^*$ \\
Enemies Aversion & \cmark$^*$ & \cmark \\
\textbf{Bottom Responsive} &  \xmark & \cmark\\
\textbf{Anonymous}  &\cmark & \xmark\\
\bottomrule
\end{tabular}\caption{A summary of the learnability and stabilizability landscape discussed in Section~\ref{sec:new_classes}. Entries marked by an asterisk symbol are consequences of previous work.}\label{table:summary}
\end{table}

We are now ready to state the following theorem, summarizing our results for the just defined HGs classes.

\begin{theorem}\label{thm:new_classes}
The results described in
Table~\ref{table:summary} hold.
\end{theorem}
\begin{proof}
We give here just a sketch of the proof. 
while the full version can be found in Appendix~\ref{apx:new_classes}.
\proofpart{Friends and Enemies.}
The efficient PAC learnability of both Friends Appreciation and Enemies Aversion profiles follows directly by observing that they are both subclasses of Additively Separable HGs (see \cite{DimitrovBHS06}), known to be efficiently PAC learnable by the results of~\citet{SliwinskiZ17}.

For what concerns stabilizability, \citet{suzuki2010hedonic} showed that Friends and Enemies Games under Friends Appreciation are a subclass of Top Responsive HGs. \citet{SliwinskiZ17} proved that Top Responsive HGs are efficiently PAC stabilizable, which then implies the same for Friends Appreciation.
For Enemies Aversion, \citet{DimitrovBHS06} prove that core-stable partitions always exist, while \citet{DimitrovS04} provide an algorithm returning such a partition. Inspired by their algorithm, we provide an algorithm PAC stabilizing this class.

\proofpart{Bottom Responsive.}
To show that the class is not efficiently PAC learnable, we prove that its pseudo-dimension is lower bounded by $2^{\frac{n-1}{2}}$, and thus is exponential in $n$. The result then follows by Theorem~\ref{thm:PAC_learnability}. The construction in our proof bears similarities to the one of \citet{SliwinskiZ17} for Top Responsive HGs.

Regarding stabilizability, we first observe that \citet{suzuki2010hedonic} show that a core-stable coalition structure always exists for this class. Moreover, a simple necessary condition for $S$ to be part of a core-stable partition $\pi$, is that for all $i \in S$ it must hold that $\set{i} \in Av(i, S)$. Indeed, if this condition is not satisfied, at least one player prefers to deviate to a singleton. To give a viable alternative for checking the condition while knowing the values of the singletons, we prove the following property:
Given a Bottom Responsive HG $H=(N, v)$, for every $i\in N$ and every $S \in \mathcal{N}_i$, it holds that $
  \{ i \} \in Av(i, S) \iff  v_i(\set{i}) \leq v_i(S)$. 
Starting from this property, we construct Algorithm~\ref{alg:bottom_responsive} which, given a sample, returns a coalition structure that is not core-blocked by any coalition from the sample. By Theorem~\ref{thm:PAC_stabilizability} this is sufficient for concluding the efficient PAC stabilizability.

\proofpart{Anonymous.}To show efficient PAC learnability, we prove that the pseudo-dimension of this class is upper bounded by $n (1+\log n)$, and thus is polynomial. Then, for each $i \in N$, the following procedure computes a hypothesis $v_i^*$ consistent with the sample in time polynomial in $n$ and $m$: For every coalition $C$ of size $k\in[n]$, if there exists $S_j$ s.t.\ $i\in S_j$ and $|S_j|=k$, then set $v_i(C)=v_i(S_j)$, otherwise set $v_i(C)=-\infty$.

For what concerns stabilizability, we can provide a counter-example showing that the class is not PAC stabilizable, even in the case of natural single-peaked preferences, where every player has a given preferred size, and the valuation decreases as the distance from such size increases.
\end{proof}

Notice that, according to the above theorem, the negative results on stabilizability of Additively Separable HGs of \cite{SliwinskiZ17} do not transfer to Friends and Enemies Games.
Furthermore, while the Bottom Reponsive HGs class is not PAC learnable but efficiently PAC stabilizable, exactly the opposite holds for Anonymous HGs. 

\begin{algorithm}[tb]
\SetNoFillComment
\DontPrintSemicolon

\caption{Stabilizing Bottom Responsive HGs}
\label{alg:bottom_responsive}
\KwIn{ $N$, $\mathcal{S}=\{( S_j,\vec{v}(S_j))\}_{j=1}^m$}
\KwOut{ $\pi$: an $\varepsilon$-stable partition of $N$}
$\pi \gets \varnothing$, $\mathcal{T} \gets \varnothing$\\
\For{$\langle S, \vec{v}(S) \rangle \in \mathcal{S}$}{
 $f \gets 1$\\
\lIf{$\exists i \in S$ s.t. $v_i(S) < v_i(\set{i})$}{
$f \gets 0$
\textbf{break}}
\lIf{f=1}{$\mathcal{T} \gets \mathcal{T} \cup \{S\}$}
\While{$\mathcal{T} \neq \varnothing$}{
$T^+ \gets \argmax_{T \in \mathcal{T}}{\modulus{T \setminus \bigcup_{P \in \pi} P}}$ \\
$\pi \gets \pi \cup \left(T^+ \setminus \bigcup_{P \in \pi} P \right)$\\
$\mathcal{T} \gets \mathcal{T} \setminus \set{T^+}$\\
$N \gets N \setminus T^+$\\
}
\lFor{$i \in N$}{
$\pi \gets \pi \cup \{\set{i}\}$
}
}
\KwRet{$\pi$}
\end{algorithm}

\section{A General Framework for Learnability: Hedonic Coalition Nets}\label{sec:learnability}
To provide a general unifying framework for learnability of HGs, a direction worth investigating is the one of determining a suitable superclass or a small number of superclasses encompassing all the learnable HG classes. Such results would contribute to the general understanding of the crucial properties leading to learnability,  or lack thereof, and would also provide means to easily determine whether a specific class of HGs is learnable.

A universal HGs class, maintaining the full expressiveness for representing any HG, is the one of the so-called \emph{Hedonic Coalition Nets} (HCNs)~\cite{Elkind09}. Before giving the definition, we note that, since there exist classes of HGs that are not learnable, it is not possible to get a positive result for the learnability of any fully expressive HGs representation, so not for HCNs either, without imposing further restrictions. Thus, our goal here is to determine suitable restrictions allowing for efficient learnability.

\begin{definition}
A \emph{hedonic coalition net (HCN)} is a tuple $(N, R_1,\ldots, R_n)$ where $N$ is a set of \emph{variables} (each corresponding to a player) and $R_i$ is the set of \emph{rules} for player $i$. A single rule in $R_i$ is given by a pair $(\phi, \beta)$, where $\phi$ is a formula of propositional logic over $N$ and $\beta \in \mathbb{R}$ is a real number. We will denote a rule in $R_i$ by $\phi \mapsto_i \beta$. Then, assuming the conventional semantic satisfaction relation ``$\models$'', the valuation of player $i$ for a coalition $S \in \mathcal{N}_i$ is
\begin{equation}\label{eq:hcn}
    v_i(S) = \sum_{\substack{\phi^j \mapsto_i \beta^j \in R_i :\\ S \models \phi^j}} \beta^j .
\end{equation}
\end{definition}

The first HCNs subclass we consider comprises HCNs in which the formulas appearing in each set of rules $R_i$ are known a priori. Namely, for any rule $\phi \mapsto_i \beta \in R_i$ we only need to learn $\beta$.
We first show that, in this case, the pseudo-dimension depends on the number of rules.
\begin{restatable}{proposition}{hcnPdim}\label{lem:hcn_pdim}
Let $\hclass(R_i)$ be the class of valuation functions that can be expressed with a fixed set of a priori known distinct rules $R_i$. Then,
$P_{dim}\left(\hclass(R_i)\right) = O\left(\modulus{R_i}\right)$.
\end{restatable}
\begin{proof}
Let $r=\modulus{R_i}$. We will show that no set of size $r+1$ can be pseudo-shattered by $\hclass(R_i)$. As a consequence $P_{dim}\left(\hclass(R_i)\right) \leq r$, which implies the result.
Let $\sample = \{S_j\}_{j=1}^{r+1}$ be any set of coalitions from $\mathcal{N}_i$ of size $r+1$, and $(t_1, \ldots, t_{r+1})$ any sequence of $r+1$ real numbers. Given any labeling $l$, the condition $v_i(S_j) > t_j \Leftrightarrow \ell_j = 1$ can be written as a system of $r+1$ linear inequalities of the form
\[
    \sum_{k=1}^r a_{jk} \beta_k > t_j \text{ if } \ell_j=1 \text{,\; and \;}
    \sum_{k=1}^r a_{jk} \beta_k \leq t_j \text{ if } \ell_j=0,
\]
where $a_{jk}= 1$ if $S_j \models \phi^k$ and $0$ otherwise.
This is a system of $r+1$ inequalities with $r$ unknowns $\beta_1,\dots, \beta_r$, thus the coefficient matrix $A=(a_{jk})$ must have linearly dependent rows. Let us w.l.o.g.\ assume that the last row $A_{r+1}$ can be written as $A_{r+1} = \sum_{j=1}^r y_j A_j$ where the coefficients $y_j$ are not all null.
Let us define the labelings $\ell^{(1)}, \ell^{(2)}$ in this way: $\ell^{(1)}_j = 1 \Leftrightarrow y_j<0$ for $j\in[r]$, $\ell^{(1)}_{r+1}=1$ and $\ell^{(2)}_j = 0 \Leftrightarrow \ell^{(1)}_j =1$.
By contradiction, assume that there exist solutions $\vec{b}_1$ and $\vec{b}_2$ that satisfy the respective systems of inequalities. Let us consider the first system. By definition of $\ell^{(1)}$ and $\vec{b}_1$, if $\ell^{(1)}_j < 0$ then $\left(A_j \cdot \vec{b}_1\right)> t_j$ but $y_j <0$ implying that $y_j \left(A_j \cdot \vec{b}_1\right) < y_j t_j$. When $\ell^{(1)}_j = 0$, instead, it holds that $y_j \left(A_j \cdot \vec{b}_1\right) \leq y_j t_j$. We can then conclude that this last inequality holds for all $j \in [r]$. Regarding $\ell^{(2)}$, with the same line of reasoning one can prove that $y_j \left(A_j \cdot \vec{b}_2\right) \geq \sum_{j=1}^r y_j t_j$ for all $j \in [r]$. Writing $A_{r+1}$ as a combination of the other rows, and including the inequalities associated to $S_{r+1}$, we obtain the following:
\begin{align*}
    t_{r+1} < A_{r+1} \cdot \vec{b}_1 &= \sum_{j=1}^r y_j \left(A_j \cdot \vec{b}_1\right) \leq \sum_{j=1}^r y_j t_j\\
    t_{r+1} \geq A_{r+1} \cdot \vec{b}_2 &= \sum_{j=1}^r y_j \left(A_j \cdot \vec{b}_2\right) \geq \sum_{j=1}^r y_j t_j
\end{align*}
implying $t_{r+1} < \sum_{j=1}^r y_j t_j \le t_{r+1}$, a contradiction.
\end{proof}

We say that $\hclass$ admits a \emph{compact} HCN representation if it is possible to represent every $v \in \hclass$ with a polynomial number of rules for each player $i$.
Observe that so far every class that has been shown to be learnable, also admits a compact HCNs representation. \citet{Elkind09} give HCNs representations for Additively Separable, Anonymous, $\W$ and $\mathcal{B}$-games.
We describe these representations and provide one for Fractional HGs in Appendix~\ref{apx:hcn_representations}.
%

The following result shows that HGs admitting a compact HCN representation, for which we know the formulas a priori, are efficiently PAC learnable.

\begin{restatable}{theorem}{simpleHCNLearnable}\label{thm:simple_hcn_learnable}
Let $\hclass$ be a class of HGs that admits a compact HCN representation. Suppose that for every set of rules $R_i$, we know the corresponding set of formulas $\Phi$. Then, $\hclass$ is efficiently PAC learnable.
\end{restatable}
%
%
While the class presented above includes Additively Separable, Fractional, and Anonymous HGs, which have all been shown to be efficiently PAC learnable, there exist other learnable classes which do not fall within the above characterization.
Indeed, for $\W$-games and $\mathcal{B}$-games, knowing the $\phi$ for each rule a priori is not possible, since the formulas themselves depend on the ordered preferences that we need to learn. On the other hand, the maximum number of distinct coalition values in both cases is only $n$.

To capture these remaining classes of learnable HGs through another suitable subclass of HCNs, we resort to \emph{decision lists}, which were introduced by \citet{Rivest87} as alternative representations for Boolean functions.
\begin{definition}
A \emph{decision list (DL)} $L$ is defined by a set of $l$ rules $L=\{(\kappa_1, b_1), \ldots, (\kappa_l, b_l)\}$ such that $\kappa_i$ is a conjunction of literals, $b_i \in \{0,1\}$, $\forall i\in [l]$, and $\kappa_l$ is the constant function \textbf{true}. Given $L$ and a truth assignment $x$, $L(x)$ is equal to $b_j$ where $j$ is the least index such that $\kappa_j(x)=1$.

We use the term $k$-decision lists ($k$-DL) if all the conjunctions in the DL are of size at most $k$.
\end{definition}
For our purposes, for any coalition $S$ and $i\in [l]$, $\kappa_i(S)=1$ if $S\models \kappa_i$ and $\kappa_i(S)=0$ otherwise. It is convenient to think of a DL as an ``\textbf{if} -- \textbf{then} -- \textbf{else if} -- ... - \textbf{else} --'' rule.
\begin{align*}
    \text{\textbf{if} } \kappa_1(S)=1 &\text{ \textbf{then return} } b_1\\
    \text{\textbf{else if} } \kappa_2(S)=1 &\text{ \textbf{then return} } b_2\\
    &\ldots\\
    \text{\textbf{else} } &\text{ \textbf{return} } b_l
\end{align*}
Note that all the HGs classes we mentioned, other than Anonymous HGs 
(see Appendix~\ref{apx:anonymous_HCN}),
can be represented as $k$-DL with $k$ constant.
It is known that for constant $k$, $k$-DL 
are efficiently PAC learnable: ~\citet{Rivest87} shows an efficient learning procedure $\learnkDL(k,\sample)$ that takes the size $k$ and a sample $\sample$ as input and returns a $k$-DL $L$ (see Appendix~\ref{apx:learnability}).
Furthermore, in the same work it is proven that $k$-DL are strictly more expressive than $k$-CNF and $k$-DNF formulas, and decision trees of depth $k$, meaning that every Boolean function that is representable in one of these forms admits a representation as a $k$-DL, but not viceversa.

Now, if we consider HCNs which contain rules that are represented by $k$-DL and additionally restrict our attention to representations in which every coalition satisfies exactly one rule, it turns out that we can again efficiently PAC learn the valuations, as shown in the following.
\begin{restatable}{theorem}{hcnDecisionLists}\label{thm:hcn_decision_list}
Let $\hclass$ be a class of HGs that admits a HCN representation such that
\begin{enumerate}[(i)]
\item every coalition $S \in \Ni$ satisfies exactly one rule in $R_i$, 
\item every rule is of the form $L \mapsto \beta$, where $L$ is a $k$-DL with $k$ constant, and $\beta$ unique, i.e., no pair of distinct rules have the same value $\beta$.
\end{enumerate}
Then, $\hclass$ is efficiently PAC learnable by Algorithm~\ref{alg:learning_nets}.
\end{restatable}

\begin{algorithm}[tb]
\SetNoFillComment
\DontPrintSemicolon
\caption{Learning HCN in $k$-DL form}\label{alg:learning_nets}
\KwIn{$k \in \mathbb{N}$, $\mathcal{S}=\{( S_j, v_i(S_j))\}_{j=1}^m$}
\KwOut{ $R_i$ consistent with $\sample$}
$R_i \gets \varnothing$\\
\For{$\beta$ in $\{v_i(S_j) : S_j\in \sample\}$}{
\For{$S_j\in \sample$}{
\leIf{$v_i(S)=\beta$} {$b_j=1$}{$b_j=0$}
}
 $\sample'=\{( S_j,b_j)\}_{j=1}^m$\\
 $L\gets \learnkDL (k,\sample')$\\
$R_i \gets R_i \cup \{L \mapsto \beta\}$
}
\KwRet{$R_i$}
\end{algorithm}

While the second assumption in Theorem~\ref{thm:hcn_decision_list} seems rather strong, we argue that asking for unique values $\beta$ actually does not impose a further restriction, even though it seems fundamental for proving the result (see Appendix~\ref{apx:learnability}).
Indeed, if there is more than one $k$-DL associated with the same value $\beta$, using the assumption that every coalition satisfies exactly one rule, it is always possible to merge them into one $k$-DL.

Theorem~~\ref{thm:hcn_decision_list} includes as a special case all HGs that can be represented by sets of mutually exclusive conjunctions, each containing at most $k$ positive literals. This is so, because we can phrase the negative literals positively within the DL, by associating the presence of such a variable with returning $0$. Thus, the conjunction size depends only on the number of positive literals. The case $k=1$ includes $\W$- and $\mathcal{B}$-games.

\section{Stabilizability of Hedonic Games}\label{sec:stabilizability}
We start this section by identifying a property that a HGs class needs to satisfy if it has any ambitions of being PAC stabilizable. To this end, we first define the set of core stable partitions w.r.t.\ a fixed sample $\sample$, and equivalence classes in a HGs class $\hclass$ w.r.t.\ a fixed sample $\sample$. Then, we state a theorem that abstracts on the arguments used in proofs showing that a specific HGs class is not PAC stabilizable. Recall that a sample $\sample$ is a set $\{( S_1, \vec{v}({S_1})), \ldots, ( S_m, \vec{v}({S_m}))\}$.

\begin{definition}
Let $\mathcal{H}$ be a class of HGs, $\sample$ a sample and let $H \in \hclass$. We denote by:
\begin{enumerate}[(i)]
    \item $C_\sample(H)= \set{\pi : \forall S\in \sample, S \text{ does not core block } \pi}$, the set of partitions consistent with the sample $\sample$;
    \item $\hclass[\mathcal{S}]$ the set of all instances $H=(N, \vec{v}\ ')\in\hclass$ such that $\vec{v}\ '(S)= \vec{v}(S)$, for each $(S, \vec{v}(S))\in\sample$.
\end{enumerate}
\end{definition}
We are now able to define the following property.
\begin{definition}
HGs class $\hclass$ satisfies the {\em sample resistant core} property, or has $\mathcal{SRC}$ in short, if for every $ \mathcal{S}\subseteq 2^N$ 
\begin{itemize}
    \item $C_\mathcal{S}(H)=\emptyset $, $\forall H \in \mathcal{H}[\mathcal{S}]$, or
    \item $\bigcap_{H \in \mathcal{H}[\mathcal{S}]} C_\mathcal{S}(H) \neq \emptyset$.
\end{itemize}
\end{definition}
\begin{restatable}{theorem}{necessaryCondPACStab}\label{thm:necessary_cond_PAC_stab}
If $\mathcal{H}$ is PAC stabilizable, then $\mathcal{H}$ has $\mathcal{SRC}$.
\end{restatable}
Notice that Theorem~\ref{thm:necessary_cond_PAC_stab} (whose proof can be found in Appendix~\ref{sec:stabilizability}) formalizes the standard approach of \citet{SliwinskiZ17} and of our work, to show that a specific HG class is not PAC stabilizable. Furthermore, property $\mathcal{SRC}$, which is of course satisfied by Top and Bottom Responsive HGs (as they can be PAC stabilized), does not seem to be a common HGs property.

One could argue that aiming to PAC stabilize a specific HG class without having any a priori knowledge on the distribution $\distribution$ is too ambitious. Thus, a natural question is whether restricting the attention to special distributions increases the prospect of stabilizing some classes of HGs. This direction was left as an open question by \citet{SliwinskiZ17} and is our focus in the remaining part of the paper. 

The motivation for limiting the scope of allowed distributions is to get a more fine-grained insight into PAC stability. The simple counterexamples from~\cite{SliwinskiZ17}, while providing valuable understandings, do not reveal ``how far away'' from achieving PAC stability certain classes of HGs are.
Thus, we proceed by studying PAC stability under a class of distributions that excludes the usual adversarial examples. In particular, we focus on distributions having a fair amount of probability mass on all coalitions. 

\begin{definition}\label{def:bounded_dist}
A distribution $\distribution$ is said to be \emph{bounded} if there exists $\lambda \geq 1$ such that, for every two coalitions $S_1, S_2$, it holds that $\Prob{S_1}{\distribution}\leq \lambda \Prob{S_2}{\distribution}$.
\end{definition}
Observe that the uniform distribution is a special case of the above definition, obtained by setting $\lambda=1$. 

A useful property that we will use extensively in our calculations is that, if $\distribution$ is bounded with a factor $\lambda$, then 
\begin{equation}\label{eq:bounded}
    \frac{1}{\lambda 2^n}\le\frac{1}{\lambda(2^n-1)}\leq\Prob{S}{S\sim\distribution}\leq \frac{\lambda}{2^n-1} \ .
\end{equation}
These simple bounds follow from the definition of a bounded distribution and the fact that $\sum_{T\in 2^N} \Prob{T}{\distribution}=1$, where the sum goes over the $2^n-1$ non-empty coalitions.
By Equation~(\ref{eq:bounded}), every coalition now has a positive probability of being sampled. Since the counterexamples to PAC stabilizability of specific HGs classes usually rely on ad hoc distributions where most of the coalitions are never sampled, this feature provides hope of obtaining better results.

\subsection{$\mathcal{W}$-games under Bounded Distributions} \label{sec:w_games}
As a case study we consider $\mathcal{W}$-games with no ties. This class admits a polynomial algorithm for finding a core stable partition~\cite{CechlarovaH04}, but, despite that, it has been shown not to be PAC stabilizable~\cite{SliwinskiZ17}. Thus, it seems a natural first candidate for being PAC stabilizable under bounded distributions. In the rest of this subsection, we indeed show the following result.

\begin{theorem}\label{thm:WgamesBoundedDistributions}
$\W$-games under bounded distributions are efficiently PAC stabilizable.
\end{theorem}

To this end, in what follows, when focusing on a fixed player $i$, w.l.o.g.\ we assume that the other players are ordered such that $v_i(1) < v_i(2) < \ldots < v_i(n-1)$. We start by exploiting the fact that the distribution is bounded.

\begin{restatable}{lemma}{lemABj}\label{lem:ABj}
Let $\eps>0$ be fixed. If we denote by $A^{(i)}_j$ the event that a sampled coalition $S$ satisfies $i,j \in S$ and $S\setminus \{i, j\} \subseteq \{j+1, \ldots, n-1\}$, and by $B^{(i)}_j$ the event that a sampled coalition $S$ satisfies $i,j \in S$ and $S\setminus \{i,j\} \subseteq \left\{\Wthreshold+2, \ldots, n-1\right\}$, it holds that 
\begin{align*}
 \Prob{A^{(i)}_j}{S\sim \distribution} \geq \frac{\eps}{2 \lambda} &\; \text{\; for \;} 1 \leq j \leq \Wthreshold, \text{ and}\\
 \Prob{B^{(i)}_j}{S\sim \distribution} \geq \frac{\eps}{4 \lambda} &\; \text{\; for \;} j > \Wthreshold.
\end{align*}
\end{restatable}
The proof of Lemma~\ref{lem:ABj} and all the others missing proofs of this subsection can be found in Appendix~\ref{sec:stabilizability}).
\citet{SliwinskiZ17} presented a simple procedure, that we will refer to as \learnWGames, which takes in input the set of players and a sample, and returns a consistent estimate $\vec{v^*}$ for the players' valuations in $\W$-games. This procedure sets $v^*_{i}(j)$ to be the $\max_{S \in \mathcal{S}_{ij}} v_i(S)$ where $\mathcal{S}_{ij} = \set{S\in \mathcal{S}: \set{i,j} \subseteq S}$, if $\mathcal{S}_{ij}$ is non-empty, $-\infty$ otherwise.
Next, we define what we call an $\eps$-estimate of a function, and show that the output of \learnWGames\ is actually such an estimate.

\begin{definition}
Function $v_i'$ is an $\eps$-\emph{estimate} of $v_i$ if
\begin{equation}
\begin{cases*}
v_i'(j) = v_i(j) & for \;$1 \leq j \leq \Wthreshold$, and\\
v_i'(j) > v_i(\Wthreshold) & for \;$j > \Wthreshold$.\\
\end{cases*}
\end{equation}
\end{definition}

\begin{restatable}{proposition}{approxWgames}\label{prop:approx_wgames}
Let $\eps, \delta >0$ and $\sample$ be a sample of size $m$. If $m\geq \frac{2\lambda}{\eps}\log{\frac{n^2}{\delta}}$, \learnWGames\ returns an $\eps$-estimate $\vec{v}^*$ of $\vec{v}$ with confidence $1-\delta$. 
\end{restatable}

\begin{algorithm}[tb]
\SetNoFillComment
\DontPrintSemicolon

\caption{Stabilizing $\W$-games}
\label{alg:partition_wgames}
\KwIn{$N$ players, $\mathcal{S}=\{( S_j,\vec{v}(S_j))\}_{j=1}^m$, $\eps >0$}
\KwOut{A partition $\pi$}
$\vec{v}^* \gets$ \learnWGames($N$, $\sample$)\\
$\pi \gets \varnothing$\\
\While{$N \neq \varnothing$}{
Pick $i \in N$ \\
\leIf{$N\setminus \{i\} \neq \varnothing$}{\\
$j \gets \argmax_{k \in N \setminus \{i\}}{v_i^*(k)}$ \\
$\pi \gets \pi \cup \{\{i,j\}\}$, 
$N \gets N\setminus \{i, j\}$\\
}{
 $\pi \gets \pi \cup \{\{i\}\}$
}
}
\KwRet{$\pi$}
\end{algorithm}

We want to show that, by relying on the $\epsilon$-estimate given by \learnWGames\ for $\epsilon$ ``not too small'', Algorithm~\ref{alg:partition_wgames} returns an $\eps$-stable partition. We first state a technical lemma.

\begin{restatable}{lemma}{greenNodes}\label{lem:green_nodes}
Let $\pi$ be the output of Algorithm~\ref{alg:partition_wgames} and let us call a player $i$ \emph{green} if it is not in a coalition with one of his $\Wthreshold$ least preferred choices according to $v_i^*$. Then,
\begin{enumerate}[a)]
    \item for $i$ green, $\Prob{i \in S \land v_i(S)>v_i(\pi(i))}{S \sim \distribution} < \lambda\eps$,
    \item for $\eps \geq \sqrt[3]{\frac{\lambda^{2}}{2^{n}}}$, $\Prob{S \text{ does not contain green } i}{S \sim \distribution} < \eps$.
\end{enumerate}
\end{restatable}
We are now finally ready to prove the main theorem, stated in the beginning of this subsection.
\begin{proof}[Proof of Theorem~\ref{thm:WgamesBoundedDistributions}]
 For $\delta>0$ and $\eps\ge \sqrt[3]{\frac{\lambda^5}{2^{n-3}}}$, we call Algorithm~\ref{alg:partition_wgames} with $\eps'=\eps/2\lambda$ and  $m\geq \frac{2\lambda}{\eps'}\log{\frac{n^2}{\delta}} = \frac{1}{\eps}\log \frac{n^2}{\delta}$, and obtain an $\eps$-stable partition with probability at least $1-\delta$. By point a) of Lemma~\ref{lem:green_nodes}, a green node has probability $<\lambda\epsilon'=\eps/2$ to get a better outcome by moving from $\pi(i)$ to $S$.
 Furthermore, $\eps'$ satisfies the requirement of point b) in Lemma~\ref{lem:green_nodes}, so the probability of sampling an $S$ without a green node is $< \eps'\le \epsilon/2$. In conclusion, if we call $G$ the event that $S$ contains a green player, then since $\Prob{S \text{ core blocks } \pi}{S\sim \distribution}$ equals
\begin{align*}
    &\Prob{S \text{ core blocks } \pi \;|\; G}{S\sim \distribution} + \Prob{S \text{ core blocks } \pi \;|\; \overline{G}}{S\sim \distribution},
\end{align*}
we see that $\Prob{S \text{ core blocks } \pi}{S\sim \distribution}\leq \epsilon/2 + \epsilon/2 = \eps$.

For $\delta>0$ and $\epsilon< \sqrt[3]{\frac{\lambda^5}{2^{n-3}}}$, with a sample of size $m\ge \frac{8\lambda^6}{\epsilon^3}\log \frac{n^2}{\delta}$, we can reveal the exact valuation functions with probability at least $1-\delta$ and return the core stable partition $\pi$ using the algorithm of \citet{CechlarovaH04}. Indeed, since the probability of drawing any coalition is by Equation~(\ref{eq:bounded}) at least $1/\lambda 2^n$, this also holds for the coalition containing only agents $i$ and $j$, which provides both $v_i(j)$ and $v_j(i)$. The probability of not drawing a particular coalition of size $2$ is $\le \left( 1 - 1/\lambda 2^n\right)^m \le e^{- m/\lambda 2^n}\le \delta/ n^2$. Taking a union bound over all the ${n\choose 2}< n^2$ coalitions of size $2$, we see that the probability of not seeing all the exact valuations is upper bounded by $\delta$.
\end{proof}

\subsection{Barriers to the Restricted Distributions Approach}

Encouraged by the positive results of the last subsection, one could try to extend the approach of focusing on bounded distributions in the hope that other classes that are known not to be PAC stabilizable, such as Additively Separable, Fractional, and Anonymous HGs, are in fact stabilizable under such distributions. Unfortunately, this does not seem to be always the case, as we discuss below.

\begin{definition}For a HGs class $\mathcal{H}$ , let $T(\mathcal{H})$ be the time complexity of the best algorithm solving the core for this class, i.e., the runtime of the fastest algorithm that for every input instance either correctly replies that the core is empty or returns a core-stable partition.
\end{definition}

\begin{restatable}{theorem}{ComputationalNotStab}\label{thm:computational_not_PAC_stabilizable}
If $T(\mathcal{H}) = \omega(poly(2^n))$ for a HG class $\mathcal{H}$, then $\mathcal{H}$ is not efficiently PAC stabilizable, even under the uniform distribution. 
\end{restatable}

Notice that the assumption $T(\mathcal{H})\in \omega(poly(2^n))$ is not that strong. In most of the HGs classes the complexity of the problem of deciding the existence of the core is either $\Sigma_2^p$-complete or NP-hard. Although this does not imply that it is not possible to find a $O(poly(2^n))$ algorithm (in the case of the total collapse of the polynomial hierarchy, this would even be possible in polynomial time), such algorithms are currently not known and at this point it seems that finding them is unlikely. In particular, the brute force approach, that searches for an element in the core by examining all the possible partitions, has a running time of $\Omega((n/2)^{n/4})$ (as this is one possible lower bound on the Bell number), and thus its running time is also in $\omega(poly(2^n))$.

\section{Conclusions}
In this work, we initiated the study of a unified approach for determining the learnability and stabilizability of specific HGs classes. One of the obvious goals for future work is finding a unique characterization of HCN representations that imply learnability. Another one is exploring further consequences of Theorem~\ref{thm:computational_not_PAC_stabilizable} and expanding the knowledge on the exact computational complexity of solving the core for the different classes of HGs.

\section*{Acknowledgements}
Giovanna Varricchio was supported by DFG grant Ho 3831/5-1. 
The authors would also like to thank the anonymous reviewers for their comments and suggestions.
\bibliography{references.bib}
\clearpage
\appendix
\section{Proof of Theorem~\ref{thm:new_classes}}\label{apx:new_classes}
In the following, we report the complete proof of Theorem~\ref{thm:new_classes} sketched in Section~\ref{sec:new_classes}. We will show that:
\begin{enumerate}[(i)]
\item Friends and Enemies Games are efficiently PAC learnable and stabilizable;
\item Bottom Responsive HGs are not PAC learnable but efficiently stabilizable, if we require as a baseline to know the values $v(i)$ for each $i \in N$;
\item Anonymous HGs are efficiently PAC learnable but not stabilizable, even in the case of natural single-peaked preferences.
\end{enumerate}
\subsection{Friends and Enemies Games}
Since the learnability for both Friends Appreciation and Enemies Aversion profiles and the stabilizability for Friends Appreciation are a consequence of previous results, in the following we will focus on proving efficient PAC stabilizability for Enemies Aversion profiles.

For this type of preferences, \citet{DimitrovBHS06} proved that core stable partitions always exist. An algorithm for computing a core stable partition was provided by~\citet{DimitrovS04}, which we describe here only informally. Initially all the agents are ``unassigned''. At every step a new coalition is created, consisting of the maximum subgroup of still unassigned agents that all consider each other friends. Once a new coalition is created, its members are marked as ``assigned''. This process continues as long as there are unassigned agents. Singleton coalitions may be created once no larger set of agents forms a friendship clique. The just described algorithm clearly has an exponential running time. In fact, due to its close relation to the \textsc{MaxClique} problem, computing a core stable partition in this setting is known to be NP-hard~\cite{DimitrovBHS06}.

Here, we present Algorithm~\ref{alg:enemy_stability} which PAC stabilizes this class.
Our algorithm, inspired by the one of \citet{DimitrovS04}, stores in $\mathcal{T}$ the sampled coalitions that are cliques (lines 2-11). Then, it repeatedly extracts a maximum clique from $\mathcal{T}$ and creates a corresponding coalition (lines 12-17). Every time a coalition is formed, since any subset of a clique is a clique, the algorithm refines the sets in $\mathcal{T}$ by removing agents which have been assigned to a coalition. This process takes place as long as $\mathcal{T}\neq\emptyset$. If there are agents which have not been assigned to any coalition, they are placed into singletons (lines 18-20).

Let $H=(N,v)$ be a Friends and Enemies Game under Enemy Aversion. To prove that Algorithm~\ref{alg:enemy_stability} stabilizes $H$, it is enough to prove that the partition $\pi$ output by the algorithm cannot be core blocked by any of the coalitions inside the sample. Then the result follows by Theorem~\ref{thm:PAC_stabilizability}. To this end, observe that any coalition $S$ from the sample ends up in the candidate set $\mathcal{T}$ if and only if all the members of the coalition consider each other to be friends. Furthermore, the all-friends coalitions are added to $\pi$ by cardinality, starting from the largest. Now, first notice that a coalition $S$ that is not an all-friends coalition, cannot block $\pi$. Furthermore, for an all-friends coalition $S$ from the sample, it is not possible that $|S|>\pi(i)$ for all $i\in S$, concluding the statement.

Notice that the exact valuation function values are not important for Algorithm~\ref{alg:enemy_stability}, since any valuation function that respects Enemy Aversion will have a negative value as soon as the coalition contains at least one enemy relation.
\begin{algorithm}[htb]
\caption{Stabilizing Enemy Aversion}\label{alg:enemy_stability}
\textbf{Input:} $N$, $\mathcal{S}=\{( S_j,\vec{v}(S_j))\}_{j=1}^m$\\
\textbf{Output:} $\pi$: an $\varepsilon$-stable partition of $N$

\begin{algorithmic}[1]
\STATE $\pi \gets \varnothing$, $\mathcal{T} \gets \varnothing$
\FOR{$\langle S, \vec{v}(S) \rangle \in \mathcal{S}$}
      \STATE $f \gets 1$
     \IF{$\exists i \in S$ s.t. $v_i(S) < 0$}
    \STATE $f \gets 0$
    \STATE \textbf{break}
    \ENDIF
     \IF{f=1} 
     \STATE $\mathcal{T} \gets \mathcal{T} \cup \{S\}$
     \ENDIF
\ENDFOR
\WHILE{$\mathcal{T} \neq \varnothing$}
    \STATE $T^+ \gets \argmax_{T \in \mathcal{T}}{\modulus{T \setminus \bigcup_{P \in \pi} P}}$
    \STATE $\pi \gets \pi \cup \left(T^+ \setminus \bigcup_{P \in \pi} P \right)$
    \STATE $\mathcal{T} \gets \mathcal{T} \setminus \set{T^+}$
    \STATE $N \gets N \setminus T^+$
\ENDWHILE    
    \FOR{$i \in N$}
    \STATE $\pi \gets \pi \cup \{\set{i}\}$
    \ENDFOR
    \RETURN $\pi$
\end{algorithmic}
\end{algorithm}

\subsection{Bottom Responsive HGs}
\proofpart{Non-Learnability.}
Consider the class $\mathcal{R}$ of HGs such that, for $S, T \in \mathcal{N}_i$, $|S| < |T| \Rightarrow v_i(S) > v_i(T)$.
First, notice that class $\mathcal{R}$ satisfies bottom responsiveness. Indeed, as for each $S \in \mathcal{N}_i$ it holds that $Av(i, S) = \{S\}$, it is easy to check that conditions (i) and (ii) hold.

Let now $\mathcal{X}=\{ S \in \mathcal{N}_i : |S| = \ceil*{\frac{n}{2}} \}$. 
In particular, for $S_1,S_2\in\mathcal{X}$, it holds that $Av(i,S_1)\cap Av(i,S_2)= \{S_1\}\cap \{S_2\}=\emptyset$.
The cardinality of $\mathcal{X}$ is 
\[
|\mathcal{X}|=\binom{n}{\ceil*{\frac{n}{2}}} > \binom{n}{\frac{n-1}{2}} >
\binom{n-1}{\frac{n-1}{2}} > 2^{\frac{n-1}{2}} 
\]
where the last inequality follows from the fact that
$
2^m = \sum_{k=1}^m \binom{m}{k} < \sum_{k=1}^m\binom{m}{k}\binom{m}{m-k} = \binom{2m}{m}.
$

In what follows, we show that $\sample = \{(S, r(S))\}_{S \in \mathcal{X}}$ is pseudo-shattered by $\mathcal{R}$, for any sequence of real numbers $\{r(S)\}_{S \in \mathcal{X}}$.
Indeed, given any labeling $\ell:\mathcal{X} \rightarrow \{0, 1\}$, we are free to choose the $v_i(S)$ in such a way that $v_i(S) > r(S) \Leftrightarrow \ell(S) = 1$ is satisfied, as long as we respect the following:
\begin{align*}
\min_{T: |T| < \ceil*{\frac{n}{2}}}{v_i(T)} &> \max_{S \in \mathcal{X}} v_i(S),\\
\max_{T: |T| > \ceil*{\frac{n}{2}}}{v_i(T)} &< \min_{S \in \mathcal{X}} v_i(S).
\end{align*}
As a consequence $P_{dim}(\mathcal{R})>|\mathcal{X}|>2^{\frac{n-1}{2}}$, 
which implies the result by Theorem~\ref{thm:PAC_learnability}.
\proofpart{Stabilizability.}
Let $H=(N, v)$ be a Bottom Responsive HG. First we will prove that the following condition holds for every $i\in N$ and $S \in \mathcal{N}_i$.
\begin{equation}\label{eq:br_av_condition}
  \{ i \} \in Av(i, S) \iff  v_i(\set{i}) \leq v_i(S) .  
\end{equation}
One direction follows by the definition of the avoid set.
Now, assume that $v_i(\set{i}) \leq v_i(S)$. By contradiction, let us assume that $\set{i} \notin Av(i, S)$ and let $T \in Av(i,S)$. By the definition of the avoid set, we know that $v_i(T) < v_i(\set{i})$. Since this holds for every $T\in Av(i,S)$ and the game is bottom responsive, since $Av(i, \set{i})=\set{\set{i}}$, by (i) of Definition~\ref{def:bottom_responsive} it follows that $v_i(S) < v_i(\set{i})$, a contradiction.

Let $H=(N, \vec{v})$ a Bottom Responsive game for which we know $v(i)$ for each $i \in N$. In the following, we will show that Algorithm~\ref{alg:bottom_responsive}
PAC stabilizes $H$. Again, by Theorem~\ref{thm:PAC_stabilizability}, it is sufficient to prove that it returns a partition which is consistent with the sample i.e. that is not core-blocked by any sampled coalition.

Since we assume that the values assigned by the players to singletons are known, Algorithm~\ref{alg:bottom_responsive} uses the equivalence from Equation~(\ref{eq:br_av_condition}) to sort out every element of the sample that does not satisfy the necessary condition to be in a core-stable partition. The remaining sets after this operation are stored in $\mathcal{T}$.
Observe that any coalition $S$ core blocking $\pi$ must satisfy both $\set{i}\in Av(i, S)$ and $\modulus{S}> \modulus{\pi(i)}$, $\forall i\in S$. The former is, as already remarked, a necessary condition for being in a core stable partition, while the latter is a consequence of this condition and the requirement (ii) in Definition~\ref{def:bottom_responsive}.
Now, consider any $S \in \mathcal{S}$.
Indeed, if $\set{i}\in Av(i, S)$  for all $i \in S$, then the algorithm adds $S$ to $\mathcal{T}$. Furthermore, the coalitions in $\pi$ are created from $\mathcal{T}$ by maximising the cardinality of the chosen set at each step, which makes it impossible for $S \in \mathcal{T}$ to satisfy $\modulus{S}> \modulus{\pi(i)}$ for all $i\in S$.

\subsection{Anonymous HGs}
Recall that a HG is said to satisfy \emph{anonimity}, as defined in~\cite{Banerjee01, Bogomolnaia02}, if $v_i(S) = v_i(T)$ for any player $i \in N$ and any $S,T\in \mathcal{N}_i$ with $|S|=|T|$. That is, players evaluate coalitions only according to their size. For this reason, we will write $s_1 \succeq_i s_2$ to denote that player $i$ prefers size $s_1$ to size $s_2$.
\proofpart{Learnability.} We will first prove that the pseudo-dimension of $v$ is bounded by $n(\log n+1)$.
Let $S_1,\dots,S_m$ be a list of subsets of $N$ and $r_1,\dots, r_m$ a list of real values. Since in Anonymous HGs the valuation of a coalition depends only on its size, we will not be able to shatter $(\langle S_1,r_1\rangle,\dots,\langle S_m,r_m\rangle)$ if the list of sets contains at least two coalitions of the same size that have a non-empty intersection. Indeed, if we w.l.o.g.\ assume that $|S_1|=|S_2|$ and $S_1\cap S_2\neq\emptyset$, we see that for any labelling in which $\ell_1\neq\ell_2$ we arrive at a contradiction. This is so, because for an agent $i\in S_1\cap S_2$ it has to hold that $v_i(S_1)=v_i(S_2)$ and at the same time $l_1\neq l_2$ implies either $v_i(S_1)>v_i(S_2)$ or $v_i(S_1)<v_i(S_2)$.
What now remains is to see that any list of sets of length at least $n(\log n +1)$ contains at least one pair of sets with the same cardinality and a non-empty intersection.
To this end, notice that there are $n$ sets of cardinality $1$ with an empty intersection, $\lfloor n/2\rfloor$ sets of cardinality $2$ with an empty intersection, and generally $\lfloor n/k\rfloor$ sets of cardinality $k$ with an empty intersection. Since, $n+\lfloor n/2\rfloor +\dots + \lfloor n/(n-1)\rfloor + \lfloor n/n\rfloor \le n + n/2 + \dots + n/(n-1) + 1 \le n(\log n +1)$, we see that the pseudo-dimension of the class of Anonymous HGs is bounded by $n(\log n+1)$.

What remains to be proven, is that we can give a consistent hypothesis in polynomial time. The following procedure computes a hypothesis $v_i^*$ consistent with the sample in time polynomial in $n$ and $m$. For every coalition $C$ of size $k\in[n]$, if there exists $S_j$ s.t.\ $i\in S_j$ and $|S_j|=k$, then set $v_i(C)=v_i(S_j)$, otherwise set $v_i(C)=-\infty$.
\proofpart{Non-Stabilizability.} The anonymity condition is often complemented by the so called \emph{single-peakedness}, where every player has a given preferred size, and the valuation decreases as the distance from such size increases. Formally:
\begin{definition}
An instance of Anonymous HGs is \emph{single-peaked} if there exists a permutation $(s_1, \ldots, s_n)$ of $\{1, \ldots, n\}$ in which every player $i\in N$ admits a \emph{peak} $p(i)$ such that $j<k \leq p(i)$ or $j>k\geq p(i)$ imply $s_k \succeq_i s_j$. Moreover if the permutation is the identity function, we say that the preference is \emph{single peaked in the natural ordering}.
\end{definition}
In the following we will prove that Anonymous HGs are not PAC stabilizable, even if we restrict our attention to single-peaked instances in the natural ordering.

Let us consider instances with $7$ agents $N=\set{a_1, a_2, a_3, a_4, b_1, c_1, c_2}$ where $a$,$b$ and $c$ are the agent types, and agents of the same type have the same preferences over coalition sizes. Consider a distribution $\distribution$ that selects uniformly at random a coalition in  $\set{S \subseteq N : \text{if } \modulus{S}=5 \text{, then } c_1,c_2\not\in S}$. In other words, the distribution never samples coalitions of size $5$ containing $c_1$ or $c_2$. In particular, this means that we will not know where coalitions of size $5$ appear in the preference list of type $c$.

Let $I_1$ be the instance in which the agents' types have the following preferences over coalition sizes:
\\
Type $a$: $6 \succ_a 5 \succ_a 4 \succ_a 3 \succ_a 2 \succ_a 1 \succ_a 7$,
\\
Type $b$: $5 \succ_b 4 \succ_b 3 \succ_b 6 \succ_b 2 \succ_b 1 \succ_b 7$,
\\
Type $c$: $4 \succ_c 3 \succ_c 5 \succ_c 6 \succ_c 2 \succ_c 1 \succ_c 7$.

The preferences in $I_1$ are single peaked in the natural ordering.
Moreover, \citet{Banerjee01} showed that instance $I_1$ has an empty core. In particular, coalitions of size $5$ containing $c_1$ or $c_2$ are never used as blocking coalitions. Therefore, a PAC stabilizing algorithm applied on $I_1$ should always report that the core is empty.

Let us now consider a slightly different instance $I_2$, obtained by changing the preferences of agents $c_1$ and $c_2$ to type $b$. Preferences in $I_2$ are still single peaked in the natural ordering, but now the partition $\pi = \set{\set{a_1,a_2}, \set{a_3, a_4, b_1, c_1, c_2}}$ is core stable.

On the one hand, since the distribution $\distribution$ never samples coalitions of size $5$ containing $c_1$ or $c_2$, no algorithm is able to distinguish between $I_1$ and $I_2$. On the other hand, $I_1$ has an empty core and $I_2$ does not. Note that in instance $I_1$, for any partition $\pi$ there exists at least one blocking coalition in the support of $\distribution$. Since the probability of selecting any coalition in the support of $\distribution$ is fixed (a positive constant), the algorithm indeed has to report that the core is empty. Thus, no algorithm can PAC stabilize this HGs class.

\section{Expressing Classes of Hedonic Games as Hedonic Coalition Nets}\label{apx:hcn_representations}
In this section, we will show how it is possible to express as HCN the classes of HGs mentioned in the paper. These representations have all been given by \citet{Elkind09}, except for the Fractional HGs one. For Anonymous HGs, only the existence of the formulas used in the rules is mentioned by \citet{Elkind09}. Writing these formulas explicitly is a non-trivial task that is, for interested readers, discussed thoroughly in Appendix~\ref{apx:anonymous_HCN}.
In the following, we consider a player $i$ and describe how to write the rules in $R_i$.
\begin{enumerate}
    \item \emph{Additively Separable}: $x_j \mapsto_i v_i(j)$, $\forall j \neq i$.
    \item \emph{Anonymous}: formulas $\phi^k, k\in[n]$ such that a subset of variables $\xi$ satisfies $\phi^k$ iff $\modulus{\xi} = k$; $\phi^k \mapsto_i v_i(k)$, $\forall k\in [n]$ are known to exist and to have polynomial length in the number of variables. For further details see Appendix~\ref{apx:anonymous_HCN}.
    \item \emph{Fractional}: assuming $\phi^k, k\in [n]$ as above; $x_j \wedge \phi^k \mapsto_i v_i(j)/k$ for $2\leq k \leq n$, $\forall j\neq i$.
    \item \emph{$\W$-games}: denote by $i_{j}$ the $j$-th player in $i$'s preference list (in descending order); $x_{n-1} \mapsto_i v_i(i_{n-1})$, $x_{n-2} \wedge \neg x_{n-1}  \mapsto_i v_i(i_{n-2})$, $x_{n-3} \wedge \neg x_{n-1} \wedge \neg x_{n-2}  \mapsto_i v_i(i_{n-3})$ and so on.
    \item \emph{$\mathcal{B}$-games}: denote by $i_{j}$ the $j$-th player in $i$'s preferences (in descending order); $x_1 \mapsto_i v_i(i_1)$, $x_2 \wedge \neg x_1  \mapsto_i v_i(i_2)$ and so on, $x_j \mapsto_i -\alpha, \forall j\neq i$, where $\alpha$ sufficiently small.
\end{enumerate}
We observe that, as stated in Section~\ref{sec:learnability}, we can assume to know the formulas a priori for the first three classes. This is so, because the rules are symmetric and we can freely associate variables to players in every $R_i$, and write $\phi$ independently from $\beta$. 

For $\W$ and $\mathcal{B}$-games, the formulas clearly depend on the ordered preferences. Thus, it is crucial that the variables $x_j$ are associated to the players in the right order (which is unknown).
\section{Expressing Anonymous Preferences as Hedonic Coalition Nets}\label{apx:anonymous_HCN}
Anonymous HGs admit a compact HCN representation, by assuming formulas $\phi^k, k\in[n]$ such that a subset of variables $\xi$ satisfies $\phi^k$ iff $\modulus{\xi} = k$. Here, we briefly and informally describe how we know that such formulas exist and why we do not state them explicitly, as for the other HGs classes. For further details, we refer the interested reader to \citet{paterson1976introduction} and \citet{wegener1987complexity}.

Let us denote by $B_n$ the set of $n$-argument Boolean functions $\{f: \{0,1\}^n \to \{0,1\}\}$. Functions in $B_n$ are to be computed by acyclic \emph{circuits} over the basis $B_2$, containing all $16$ Boolean functions with two arguments. 

An acyclic circuit may be represented as a finite directed acyclic graph with $n$ input nodes and one output node, where each input node corresponds to one of the arguments and each intermediate node is associated to an element of $B_2$. Furthermore, the indegree of the input nodes is zero, while every intermediate node has an ordered pair of incoming arcs. An immediate complexity measure of circuits is the \emph{circuit size}, $c$, which counts the number of intermediate nodes, which we refer to as \emph{logical gates}. Another parameter, motivated by the fact that if each logical gate in the circuit requires the same execution time, the bottleneck in parallel computation will be the \emph{circuit depth}, $d$, which counts the maximum number of logical gates on any path from an input node to the output node. 

Now, if we want to represent a Boolean function as a linear expression over the input variables with function symbols corresponding to elements of $B_2$, we need to build an acyclic circuit where all the logical gates have outdegree one. This transformation can be easily done by replicating inputs multiple times, as we do not longer allow for using a result of any logical gate computation more than once. Such circuits are then called \emph{formulas}, and the size of the formula, $\ell$, is equal to the number of its internal nodes. It is easy to see that for any Boolean function $f$, it holds that $c(f)\le \ell(f) \le 2^{d(f)}$.

A Boolean function is called \emph{symmetric}, if its value depends only of the number of $1$s in the input. It is known that $\ell(S_n)\in O(n)$, where $S_n\subset B_n$ is the class of all symmetric Boolean functions. The proof of this claim easily follows from a two-stage construction of the formula, in which the first stage is to construct a circuit which for any input vector $(x_i)_{i\in[n]}$ computes the binary representation of $\sum_i x_i$ and in the second stage, the required function is computed from the $\lceil \log (n+1) \rceil$ results of the first stage.
The first stage can be computed by using a recursive procedure in which the binary representations of the two halves of the argument set are first computed separately and then added together. Such a circuit has depth $O(\log n)$, and each node performs an addition of two $\lceil \log (n+1) \rceil$ bit numbers, which can naively be done with a circuit of depth $O(\log n)$ but is also possible to do with a circuit with depth $O(\log \lceil \log (n+1) \rceil)$. Furthermore, for the second stage, since it is also known that the circuit depth of $B_m$ is bounded by $m+1$, we arrive at the claimed bound, as the number of arguments in the second stage is $\lceil \log (n+1) \rceil$ and by the fact that $\ell(f)\le 2^{d(f)}$.

All the claims above were made for the logical gates corresponding to elements of $B_2$. We say that a basis $\Omega$ is complete if any Boolean function can be computed in an $\Omega$-circuit. Already some smaller bases, as for instance $\{\land, \lor, \lnot\}$, are complete, and so is the base consisting of all the conventional Boolean operators and truth constants, that we assume in HCNs.  Since the complexity and depth of Boolean functions can increase only by a constant factor when switching from one complete basis to another, we were free to restrict our attention to $B_2$.

The specific symmetric function that fits the purpose of representing Anonymous HGs is the \emph{exactly-$k$-function}, $E_k^n(x_1,\dots, x_n)=1 \iff \sum_i x_i=k$. Alternatively, as implicitly proposed by \citet{Elkind09}, one could use the \emph{threshold function}, $T_k^n(x_1,\dots, x_n)=1 \iff \sum_i x_i\ge k$, together with a construction similar to the one for $\W$-games for building the HCN representation.

There are many further known results regarding the formula size both for general and specific symmetric functions that can be used to derive the existence of a compact representation for Anonymous HGs via HCNs (see also~\citet{Dunne}), but all of them, to the best of our knowledge, include constructions that result in formulas that are much more complicated to write down than what was necessary for any of the other considered HGs classes.

Finally, we remark that it is not possible to represent Anonymous HGs as $k$-DL. To be able to recognize a coalition of size $m$, the DL cannot contain only conjunctions of size strictly less than $m$. Otherwise, the output would depend on less than $m$ variables, which is a contradiction.

\section{Omitted Proofs from Section~\ref{sec:learnability}}\label{apx:learnability}
Here we report the proofs of the results in Section~\ref{sec:learnability}, starting with Theorem~\ref{thm:simple_hcn_learnable}.
\simpleHCNLearnable*
\begin{proof}
By Proposition~\ref{lem:hcn_pdim}, we know that the pseudo-dimension of the class of valuation functions of every player is polynomial in $n$.
Now, by Theorem~\ref{thm:PAC_learnability}, if we can show that for any sample $\sample$ we can give a hypothesis $\vec{v}^*$ consistent with the sample, we obtain the result. The following procedure infers the real values $\beta^j$ to associate to each rule $\phi^j \in \Phi$. To every $S \in \sample\cap \mathcal{N}_i$, by defining $\phi^j(S)=1$ if $S\models\phi^j$ and $0$ otherwise, we associate a linear equation of the form
\begin{equation*}
    v_i(S)=\sum_{\phi^j \in \Phi} \beta^j \phi^j(S),
\end{equation*}
where $\beta^j$ are the unknown variables. We obtain a linear system given by all the equations associated to $S \in \sample\cap \mathcal{N}_i$. Depending on the rank of this system induced by the sample, we can either solve it exactly or get multiple solutions.
\end{proof}

We conclude with the proof regarding the learnability of HCNs whose formulas are $k$-DL. For completeness, we also report the procedure $\learnkDL(k,\sample)$ given by \citet{Rivest87}, which is employed by Algorithm~\ref{alg:learning_nets}.

\begin{algorithm}[htb]
\caption{\learnkDL~\cite{Rivest87}}\label{alg:learning_kDL}
\textbf{Input}: $k \in \mathbb{N}$, $\mathcal{S}=\{( S_j,b_j))\}_{j=1}^m$\\
\textbf{Output}: $L$, a $k$-DL consistent with $\sample$

\begin{algorithmic}[1] 
\STATE $L \gets \varnothing$\\
\WHILE{$\sample \neq \varnothing$}
\FOR{$\kappa \in C^n_k$}
\IF{every set in $K = \{S \in \sample : S \models \kappa\}$ has the same label $b$}
\STATE $L \gets L \cup \{(\kappa, b)\}$\\
\STATE $\sample \gets \sample \setminus K$\\
\ENDIF
\STATE \textbf{break}\\
\ENDFOR
\ENDWHILE
\STATE \textbf{return} $L$\\
\end{algorithmic}
\end{algorithm}

\hcnDecisionLists*
\begin{proof}
Let us first show that the pseudo-dimension of the class of players' valuations is polynomial in $n$. Note that in this setting we assume the exact formulas to be unknown, thus we cannot use Proposition~\ref{lem:hcn_pdim}. By (i), every coalition $S$ satisfies only one rule and its value is unique by (ii). As a consequence, if the number of rules equals $r$, then no sample of size $r(n)+1$ can be pseudo-shattered. Indeed, any sample of size $r(n)+1$ contains two coalitions with the same associated value, thus every labeling containing different labels for these coalitions cannot be satisfied. The largest set of rules which respects our assumption is the one in which every DL contains all of the possible conjunctions of size at most $k$, and only one of them returns $1$. 
In this case, the number of different rules $r\in\Theta(n^k)$, so generally it holds that $P_{dim}(\hclass)=O(n^k)$, i.e., polynomial in $n$ since $k$ is constant.

Now, let us show that Algorithm~\ref{alg:learning_nets} produces a hypothesis consistent with the sample. The learning problem in this case is to associate the correct rule to each different $\beta$. By the second assumption, every rule is represented as a $k$-DL and Algorithm~\ref{alg:learning_kDL} is known to return a $k$-DL consistent with a binary labeled sample~\cite{Rivest87}. Algorithm~\ref{alg:learning_nets}, thus, proceeds by assigning binary labels for every possible $\beta$ separately, setting the label to $1$ for coalitions whose value equals $\beta$ and $0$ for all others. Then, it calls Algorithm~\ref{alg:learning_kDL} as a subroutine for determining the DL associated to this particular value $\beta$. The correctness of Algorithm~\ref{alg:learning_nets} therefore follows from the correctness of Algorithm~\ref{alg:learning_kDL}.
\end{proof}

\section{Omitted Proofs from Section~\ref{sec:stabilizability}}
Here we report the proofs of the results in Section~\ref{sec:stabilizability} starting with Theorem~\ref{thm:necessary_cond_PAC_stab}.

\necessaryCondPACStab*
\begin{proof}
Assume that $\mathcal{H}$ does not satisfy $\mathcal{SRC}$. This means that there exists a sample $\sample^*$ such that there exist $H^1, H^2 \in \mathcal{H}[\sample^*]$ such that $C_{\sample^*}(H^1)\neq \emptyset$ and $C_{\sample^*}(H^1) \cap C_{\sample^*}(H^2)=\emptyset$. If we now define a distribution $\mathcal{D}$ with positive probability mass only on the elements of $\mathcal{S}^*$, we will be able to conclude that $\mathcal{H}$ is not PAC stabilizable. Indeed, no algorithm can distinguish between $H^1$ and $H^2$ when sampling from $\distribution$ and, thus, cannot return a correct answer on both instances.
\end{proof}

The following are the missing proofs for the part on PAC stability of $\W$-games with no ties. 

\lemABj*
\begin{proof}
Let $j \leq \Wthreshold$. By Equation~(\ref{eq:bounded}) we have that
\[
\Prob{A^{(i)}_j}{S\sim \distribution} \geq \frac{2^{n-j-1}}{\lambda 2^n} = \frac{1}{\lambda 2^{j+1}} \geq \frac{1}{\lambda2^{\Wthreshold + 1}} \geq \frac{\eps}{2 \lambda}.
\]
Similarly, if $j > \Wthreshold$, we see that
\[
\Prob{B^{(i)}_j}{S\sim \distribution} \geq \frac{2^{n-\Wthreshold-2}}{\lambda 2^n} = \frac{1}{\lambda 2^{\Wthreshold+2}} \geq \frac{\eps}{4\lambda}.
\]
\end{proof}

\approxWgames*
\begin{proof}
Consider a fixed player $i$,
let $A^{(i)}_j, B^{(i)}_j$ be the events defined as in Lemma~\ref{lem:ABj} and $\sample_{ij}$ as defined in Algorithm~\ref{alg:estimate_wgames}.
First, let $j \leq \Wthreshold$.
Observe that, if $A^{(i)}_j$ holds for at least one $S \in \mathcal{S}$, then $v_i^*(j) = v_i(j)$. Indeed, in that case $v_i^*(j) = \max_{S \in \mathcal{S}_{ij}} v_i(S) = v_i(j)$. By Lemma~\ref{lem:ABj},
\begin{equation*}
    \Prob{\overline{A^{(i)}_j}}{\mathcal{S}\sim \distribution} \leq \left( 1-\frac{\eps}{2\lambda}\right)^m \leq e^{-m\eps/2\lambda} \leq \frac{\delta}{n^2}.
\end{equation*}
Now, let $j > \Wthreshold$.
Observe that, if $B^{(i)}_j$ holds for at least one $S \in \mathcal{S}$, then $v_i^*(j) > v_i(\Wthreshold)$. Indeed, $v_i^*(j) = \max_{S \in \mathcal{S}_{ij}} v_i(S) \geq v_i(\lfloor \log_2 \frac{1}{\epsilon}\rfloor + 1)$.
Again by Lemma~\ref{lem:ABj},
\begin{equation*}
    \Prob{\overline{B^{(i)}_j}}{\mathcal{S}\sim \distribution} \leq \left( 1-\frac{\eps}{4\lambda}\right)^m \leq e^{-m\eps/4\lambda} \leq \left(\frac{\delta}{n^2}\right)^2 \leq \frac{\delta}{n^2}.
\end{equation*}
Now, let $\overline{A^{(i)}} = \bigcup_{1\leq j\leq \Wthreshold}\overline{A^{(i)}_j}$ and $\overline{B^{(i)}}=\bigcup_{j> \Wthreshold}\overline{B^{(i)}_j}$.
From the previously shown bounds,
\begin{equation*}
    \Prob{\overline{A^{(i)}} \cup \overline{B^{(i)}}}{\mathcal{S}\sim \distribution} \leq (n-1)\frac{\delta}{n^2} \leq \frac{\delta}{n},
\end{equation*}
which means that the probability that $v_i^*$ is not a $\eps$-estimate of $v_i$ is less then $\delta/n$. Considering the whole $\vec{v}$, by using a union bound, we obtain that $\vec{v}^*$ is an $\eps$-estimate of $\vec{v}$ with confidence $1-\delta$.
\end{proof}

\greenNodes*
\begin{proof}
To prove a), it is enough to notice that for green $i$ 
\begin{align*}
        &\Prob{i \in S \;\land\; v_i(S)>v_i(\pi(i))}{S \sim \distribution}\\
        \leq &\Prob{S \setminus \{i\} \subseteq \left\{ \Wthreshold+2, \ldots, n-1\right\}}{S \sim \distribution}\\
        \leq &\frac{2^{n-\Wthreshold-2}\lambda}{ 2^n-1} < \frac{\lambda}{2^{\Wthreshold + 1}} \le \lambda\epsilon.
\end{align*}

To prove b), we start by observing that at least $(n-\Wthreshold+2)/2$ players are green. Indeed, notice that at the beginning of iteration $t$ there are exactly $n-2t+2$ players left. A player picked by the algorithm is associated with the remaining player ranked the highest in their preference list. If $t < (n-\Wthreshold+2)/2$, the number of remaining nodes is $>\Wthreshold$, so the player picked at iteration $t$ is green.

Now, we arrive at the statement by observing that
\begin{align*}
    &\Prob{S\text{ does not contain green } i}{S \sim \distribution}\\
    &\leq \frac{\lambda}{2^n-1}(2^{n-(n-\Wthreshold+2)/2} -1) \\
    &< \frac{\lambda 2^{n/2-1 + \Wthreshold/2}}{2^{n-1}}
    \leq \frac{\lambda\sqrt{1/\epsilon}}{2^{n/2}} \leq \eps
\end{align*}
where the first inequality follows from the bound on the number of green players and Equation~(\ref{eq:bounded}), and the last one from the assumption on $\eps$.
\end{proof}

We conclude with the proof of the last result of the paper.

\ComputationalNotStab*
\begin{proof}
Let $\distribution$ be the uniform distribution on $2^N$ and let us assume that a PAC stabilizing algorithm exists. Given $\epsilon = \frac{1}{2^{n+1}}$, the PAC stabiliz ing algorithm must solve the core. Indeed, the PAC stabilizing algorithm cannot return a partition that is not core stable, as the probability of sampling a blocking coalition is at least $\frac{1}{2^{n}-1} > \epsilon$.
By definition, the running time of the PAC stabilizing algorithm is polynomial in $1/\epsilon$  and $n$, implying that its time complexity is $O(poly(2^n))$.
Therefore, we reached a contradiction.
\end{proof}
\end{document}